\documentclass{article}
\usepackage{amsfonts}
\usepackage{color}
\usepackage{amsmath}

\newtheorem{theorem}{Theorem}

\newtheorem{lemma}[theorem]{Lemma}

\newtheorem{proposition}[theorem]{Proposition}
\newtheorem{remark}[theorem]{Remark}

\newenvironment{proof}[1][Proof]{\noindent \textbf{#1.} }{\  \rule{0.5em}{0.5em}}

\newcommand{\bq}{\begin{equation*}}
\newcommand{\eq}{\end{equation*}}
\newcommand{\bqn}{\begin{equation}}
\newcommand{\eqn}{\end{equation}}
\newcommand{\bqq}{\begin{eqnarray*}}
\newcommand{\eqq}{\end{eqnarray*}}
\newcommand{\bqqn}{\begin{eqnarray}}
\newcommand{\eqqn}{\end{eqnarray}}

\DeclareMathOperator*{\esssup}{ess\;sup}

\newcommand{\ep}{\epsilon}
\newcommand{\p}{\partial}
\newcommand{\be}{\begin{eqnarray}}
\newcommand{\ee}{\end{eqnarray}}
\newcommand{\bee}{\begin{eqnarray*}}
\newcommand{\eee}{\end{eqnarray*}}

\begin{document}

\title{Analysis of the optimal exercise boundary of American put options with delivery lags%
\thanks{\textcolor[rgb]{1.00,0.00,0.00}{We are grateful to the anonymous referee whose comments and
suggestions have significantly improved our paper, especially
Theorem 1.1(i). We thank Mihail Zervos for introducing this topic to
us. We also thank Tiziano De Angelis for sharing with us his
insights (and a detailed proof) on the large time behavior of free
boundaries using probabilistic arguments.} This work is supported by
NNSF of China (Grant No. 11771158, 11801091) and Guangdong Basic and
Applied Basic Research Foundation (Grant No.2019A1515011338).}}
\author{Gechun Liang\thanks{%
Department of Statistics, University of Warwick, Coventry, CV4 7AL,
U.K.; email: \texttt{%
g.liang@warwick.ac.uk}}
\and Zhou Yang\thanks{%
School of Mathematical Sciences, South China Normal
 University, Guangzhou 510631, China; email: \texttt{%
yangzhou@scnu.edu.cn}}}
\date{}
\maketitle

\begin{abstract}
A make-your-mind-up option is an American derivative with delivery lags. We show that its put option can be
decomposed as a European put and a new type of American-style
derivative. The latter is an option for which the investor receives
the Greek Theta of the corresponding European option as the running
payoff, and decides an optimal stopping time to terminate the
contract. Based on this decomposition and using  free boundary techniques, we show that the
associated optimal exercise boundary exists and is a strictly
increasing and smooth curve, and analyze the asymptotic behavior of the value function and the optimal exercise boundary
for both large
maturity and small time lag .\\


\end{abstract}



\section{Introduction}

With a few exceptions, models of optimal stopping time problems
assume that the player is able to terminate the underlying
stochastic dynamics immediately after the decision to stop, or to
bring a new project online without any delays after the decision to
invest. In fact, both stopping stochastic dynamics and initiating a
new project take time.

In this paper, we consider a class of optimal
stopping problems where there exists a time lag between the
player's decision time and the time that the payoff is delivered. In particular, we study American put options with delivery lags in details. In
practice, there may exist a time lag between the time that the
option holder decides to exercise the option and the time that the
payoff is delivered. Such delivery lags may be specified in
financial contracts, where the decision to exercise must be made
before the exercise takes place. They are called
\emph{make-your-mind-up options} (see Chapter 6 of \cite{Jiang} and
Chapter 9 of \cite{Wilmott}). For example, the option holder must
give a notice period before she exercises, and she cannot change her
mind. On the other hand, even for a standard American derivative,
the option holder may not be able to exercise it immediately, when
there exist liquidation constraints in financial markets.

Let $W$ be a one-dimensional Brownian motion on a complete probability space
$(\Omega,\mathcal{F},\mathbf{P})$. Denote by
$\mathbb{F}=\{\mathcal{F}_t\}_{t\geq 0}$ the augmented filtration
generated by $W$. Let a constant $T>0$ represent the maturity and
another constant $\delta\in[0,T)$ represent the time lag. The player aims to choose an optimal stopping time
$\tau^{0,*}\in\mathcal{R}_t^{0}$ in order to maximize {the discounted expected payoff}
\begin{align}\label{optimal_stopping_delay_special_3}
Y_t^{\delta}= \esssup_{\tau^{0}\in\mathcal{R}_t^{0}}
&\mathbf{E}\left[e^{-r(\tau^{0}+\delta-t)}(K-X_{\tau^{0}+\delta})^+\mathbf{1}_{\{\tau^{0}+\delta<
T\}}\right.\notag\\
&\ \ \ \left.+\ e^{-r(T-t)}(K-X_{T})^+\mathbf{1}_{\{\tau^{0}+\delta\geq
T\}}|\mathcal{F}_t\right],\quad t\in[0,T],
\end{align}
where $Y^{\delta}$ represents the value process, the $\mathbb{F}$-adapted process $X$ models the stock price, the constant $K>0$ denotes the strike price, and
$${\mathcal{R}_t^{0}:=\{\tau^{0}:\Omega\rightarrow[t,T],\
\text{and}\  \{\tau^{0}\leq s\}\in\mathcal{F}_{s}\ \text{for\ any}\
s\in[t,T]\}.}$$

Note that $\delta=0$ corresponds to the classical optimal stopping
problem for American options (see, for example, \cite{Detemple} and \cite{Peskir})\textcolor{red}{, so we mainly focus on the problem in the case of $\delta>0$ in this paper}. For $\delta>0$, if the player decides to stop at some
stopping time $\tau^{0}$, then the payoff will be delivered at
$\tau^0+\delta$ rather than $\tau^0$, so there is a time lag of the
delivery of the payoff. We also observe that the problem (\ref{optimal_stopping_delay_special_3}) is trivial for
$t\in(T-\delta,T]$ for, in this situation, the expected payoff is independent of choice of $\tau^0$, and the player
may simply choose the optimal stopping time as the maturity $T$.
\emph{Thus, we focus on the case $t\in[0,T-\delta]$ throughout the paper.}

Although this type of optimal stopping problems with delivery lags
have been well studied in the literature (see \cite{Keppo} and
\cite{Oksendal} with more references therein), little is known about
the corresponding optimal exercise boundaries and their asymptotic
behavior for small time lag $\delta$ and large maturity $T$.
Intuitively, both the value function and the corresponding optimal
exercise boundary (if exists) will converge to the solution for the
case without delivery lags when $\delta\downarrow 0$, and to the
solution for the perpetual case when $T\uparrow \infty$. It is the
aim of this paper to prove the above asymptotic behavior using free
boundary techniques.

To be more specific, under the geometric Brownian motion setup, we prove the following result.

\begin{theorem}\label{main_theorem}
Suppose that the stock price $X$ follows
$$dX_s/X_s=(r-q)ds+\sigma dW_s,\quad X_t=X,$$
where the interest rate $r>0$, the dividend rate
\textcolor{red}{$q\in[0,r]$}\footnote{\textcolor{red}{$q\leq r$ is a
technique assumption, which ensures conclusion (i) in Proposition
4.}} and the volatility $\sigma>0$ are all constants. Then, the
following assertions hold:

(i) The value $Y^{\delta}_t=V^\delta(t,X_t)$ is decreasing with respect to
 $\delta$ and, moreover,
 \begin{equation}\label{bound}
V^{0}(t,X)\geq V^{\delta}(t,X)\geq
V^{0}(t,X)-\textcolor{red}{K\left(1-e^{-r\delta}\right),\quad
t\in[0,T-\delta]},
 \end{equation}
where $V^{\delta}(\cdot,\cdot)$ and $V^0(\cdot,\cdot)$ represent the
value function for the American put with and without delivery lags,
respectively. \textcolor[rgb]{1.00,0.00,0.00}{In addition,
$V^\delta(t,X)$ is decreasing with respect to $t$.}

(ii) There exists an optimal exercise boundary $X^{\delta}(t)\in C^\infty[0,T-\delta)$ separating exercise and continuation regions (cf. (\ref{relation}) and (\ref{freeboundary_10}))). Moreover, it is strictly increasing in $t$, with the end point $$X^{\delta}(T-\delta)=\lim\limits_{t\rightarrow(T-\delta)^{-}}X^{\delta}(t)=Ke^{\overline{X}},$$
where $\overline{X}$ is given in Proposition \ref{Pro3}.

(iii) The optimal exercise boundary $X^{\delta}(t)\rightarrow K e^{\underline{X}}$ as
$T\rightarrow\infty$ with $\underline{X}$ given in (\ref{freeboundary_0}), so $Ke^{\underline{X}}$ is the asymptotic line
of the optimal exercise boundary $X^{\delta}(t)$. Moreover,
$X^\delta(t)\rightarrow X^0(t)$ for any
$t\in[0,T)$
 as $\delta\rightarrow 0$, where $X^0(t)$ represents the optimal exercise boundary for the corresponding American put without delivery lags.
\end{theorem}

To prove Theorem \ref{main_theorem}, we need to first solve the optimal stopping problem (\ref{optimal_stopping_delay_special_3}). A basic idea is to introduce a new
obstacle (payoff) process, which is the projection (conditional expectation)
of the original expected payoff. For $t\in[0,T-\delta]$,
define  \begin{equation}\label{european}
\widehat{Y}_t^{\delta}=\mathbf{E}\left[e^{-r\delta}(K-X_{t+\delta})^+|\mathcal{F}_t\right],
\end{equation}
which is the time $t$ value of the corresponding European put option
with maturity $t+\delta$. Denote by $P(\cdot,\cdot)$ the value
function of the European put option with maturity $T$. Then, the time homogeneity of (\ref{european}) implies
$\hat{Y}_t^{\delta}=P(T-\delta, X_t)$, and the tower property of conditional expectations further yields
\begin{equation}\label{new_problem}
Y_t^{\delta}=V^{\delta}(t,X_t)=\textcolor{red}{\esssup_{\tau^{0}\in\mathcal{R}_t^{0}}
\mathbf{E}\left[e^{-r((\tau_0\wedge(T-\delta))-t)}P(T-\delta,X_{\tau_0\wedge(T-\delta)})|\mathcal{F}_t\right]}
\end{equation}
\textcolor{red}{with $x\wedge y=\min(x,y)$ and $t\in[0,T-\delta]$.}
Hence, we have transformed the original problem
(\ref{optimal_stopping_delay_special_3}) to a standard optimal
stopping problem (without delivery lags) with the European option
price as the new obstacle process. The rest of the paper will
therefore focus on (\ref{new_problem}) and its corresponding
variational inequality (\ref{VI11}) in section \ref{sec:american}.

The existing literature of optimal stopping with delivery lags
(see \cite{Keppo} and \cite{Oksendal} for example) usually assumes that
the payoff is a linear function of the underlying asset $X$, which certainly excludes the American payoff.
A consequence of this simplified assumption is that the new obstacle
$\hat{Y}^{\delta}$ is also linear in $X$, which follows from the
linearity of the conditional expectation, and the obstacle function in the
variational inequality is therefore also a linear function. Hence, the treatments of the optimal stopping problems with and
without delivery lags are essentially the same in their models.

In our case, since the American payoff is only a piecewise linear function of the
underlying asset $X$ (with a kink point at $K$), this kink point
propagates via the conditional expectation, resulting in a nonlinear
obstacle function $P(T-\delta,\cdot)$. This differentiates our problem from
the existing optimal stopping problems with delivery lags, and makes
the analysis of the corresponding optimal exercise boundary much
more challenging.



We first develop an early exercise premium decomposition
formula for the American put option with delivery lags (see (\ref{Decomposition})). This helps us
overcome the difficulty of handling the European option price as the
modified payoff. We show that an American put option with delivery
lags can be decomposed as a European put option and another
American-style derivative as an auxiliary optimal stopping problem (see (\ref{optimal_stopping_delay_special_4})). The latter is an option for which the
investor receives the Greek Theta of the corresponding European
option as the running payoff, and decides an optimal stopping time
to terminate the contract. The
decomposition formula (\ref{Decomposition}) can also be regarded as
a counterpart of the early exercise premium representation of
standard American options, and is crucial to the analysis of the
associated optimal exercise boundary.

Using free-boundary techniques, we then give a detailed analysis of
the associated optimal exercise boundary. An essential difficulty
herein is \textcolor{red}{the non-monotonicity of the difference between the value function and the payoff} with respect to
the stock price (a similar phenomenon also appears in \cite{Dai}).
As a result, it is not even clear \emph{ex ante} whether the optimal
exercise boundary exists or not. This is in contrast to standard
American options, for which the value function, subtracted by the
payoff, is monotonic with respect to the stock price, so the
stopping and continuation regions can be easily separated.

Thanks to the auxiliary optimal stopping problem (\ref{optimal_stopping_delay_special_4}) and its associated variational inequality (\ref{VI1}), we prove that the optimal exercise
boundary exists and is a strictly increasing and smooth curve, with
its end point closely related to the zero crossing point of the
Greek Theta of the corresponding European option. Intuitively, when Theta is positive, the running payoff
of the new American-style derivative is also positive, so the
investor will hold the option to receive the positive Theta
continuously. In contrast, when Theta is negative, one may think
that the investor would then exercise the option to stop her losses.
However, we show that when Theta is negative but not too small, the
investor may still hold the option and wait for Theta to rally at a
later time to recover her previous losses. We further quantify such
negative values of Theta by identifying the asymptotic line of the
optimal exercise boundary, which turns out to be the optimal
exercise boundary of the corresponding perpetual problem.

The paper is organized as follows. In section
\ref{sec:american}, we prove Theorem \ref{main_theorem} (i) and introduce the early excise premium decomposition formula. We then consider the corresponding perpetual problem in section \ref{sec:perpetual}, and in section \ref{sec:finite_horizon} we prove Theorem \ref{main_theorem} (ii) and (iii). Some technical
proofs about the property of the Greek Theta are provided in the appendix.


\section{The variational inequality characteriation}\label{sec:american}
We first solve the optimal stopping problem (\ref{new_problem}) via its associated variational inequality
\begin{eqnarray}\label{VI11}
 \left\{
 \begin{array}{ll}
 (-\p_t-{{\cal L}})V^{\delta}(t,X)=0,&\mbox{if}\; V^{\delta}(t,X)> P(T-\delta,X),\\
  &\mbox{for}\;(t,X)\in{\Omega}_{T-\delta};
 \vspace{2mm} \\
 (-\p_t-{{\cal L}})V^{\delta}(t,X)\geq 0,&\mbox{if}\; V^{\delta}(t,X)=P(T-\delta,X),\\
 &\mbox{for}\;(t,X)\in{\Omega}_{T-\delta};
 \vspace{2mm} \\
 V^{\delta}(T-\delta,X)=P(T-\delta,X),
 &\mbox{for}\;X\in\mathbb{R}_+,
 \end{array}
 \right.
\end{eqnarray}
with ${\Omega}_{T-\delta}= [0,T-\delta\,)\times\mathbb{R}_+$, and
the operator $\mathcal{L}$ given by the Black-Scholes differential
operator
$${\cal L}  ={1\over2}\,\sigma^2X^2\p_{XX}
 +(r-q)X\p_X-r. $$

Note that if $\delta=0$, $P(T-\delta,X)=(K-X)^+$, and variational
inequality (\ref{VI11}) reduces to the standard variational
inequality for American put options. On the other hand, since variational inequality (\ref{VI11}) is with
smooth coefficients and obstacle, its (strong) solution
$V^\delta(\cdot,\cdot)$ characterizes the value function and the optimal stopping rule for the optimal stopping problem (\ref{new_problem}).

\begin{proposition}\label{regularity1}
The value function $V^{\delta}(\cdot,\cdot)$ of the optimal stopping problem (\ref{new_problem}) is the unique
bounded strong solution to variational inequality~\eqref{VI11}, and the optimal stopping rule is given by
\begin{equation}\label{optimal_stopping_rule}
\textcolor[rgb]{1.00,0.00,0.00}{\tau^{0,*}=\inf\{s\in[t,T-\delta]:
V^{\delta}(s,X_s)=P(T-\delta, X_s)\}\wedge T.}
\end{equation}
Moreover, $V^\delta\in W^{2,1}_{p,loc}(\Omega_{T-\delta})\cap
C(\overline{\Omega_{T-\delta}})$ for any $p\geq 1$, and $\p_x
V^\delta\in C(\overline{\Omega_{T-\delta}})$.

Herein, $W^{2,1}_{p,loc}(\Omega_{T-\delta})$ is the set of all
functions whose restriction on any compact subset
$\Omega_{T-\delta}^*\subset\Omega_{T-\delta}$ belong to
$W^{2,1}_p(\Omega_{T-\delta}^{*})$, where
$W^{2,1}_p(\Omega_{T-\delta}^{*})$ is the completion of
$C^{\infty}(\Omega_{T-\delta}^{*})$ under the norm
$$||V^\delta||_{W^{2,1}_p(\Omega_{T-\delta}^{*})}=
\left[\int_{\Omega_{T-\delta}^*}(|V^\delta|^p+|\partial_tV^\delta|^p+|\partial_xV^\delta|^p+|\partial_{xx}V^\delta|^{p})dxdt\right]^{\frac{1}{p}}
.$$
\end{proposition}

The proof follows along the similar arguments used in
Chapter 1 of \cite{Friedman2}, or more recently \cite{Yan2}, and is thus omitted.

\subsection{Proof of Theorem \ref{main_theorem} (i)}

\textcolor{red}{In this subsection, we prove Theorem
\ref{main_theorem} (i).\footnote{We thank the referee for outlining
the current probabilistic proof for us. Note that the proof does not
require the geometric Brownian motion model of the underlying asset,
which is more general than our original proof based on PDE
arguments.} Note that the arguments below do not rely on the
geometric Brownian motion assumption on $X$, as long as its
discounted price $e^{-(r-q)t}X_t$ is a martingale.}

\textcolor{red}{We first prove the monotone property of
$V^{\delta}(\cdot,\cdot)$ with respect to $\delta$. Fix $0\leq
\delta_1<\delta_2$. For any $\tau_2\in \mathcal{R}_t^{0}$, take
$\tau_1=(\tau_2+(\delta_2-\delta_1))\wedge T$. Since
$\delta_1,\delta_2,T$ are constants, we know that $\tau_1\in
\mathcal{R}_t^{0}$. Moreover, it is easy to check that
$\{\tau_1+\delta_1\geq T\}=\{\tau_2+\delta_2\geq T\}$ and
$$
 e^{-r(\tau_1+\delta_1-t)}(K-X_{\tau_1+\delta_1})^+\mathbf{1}_{\{\tau_1+\delta_1<
T\}}=e^{-r(\tau_2+\delta_2-t)}(K-X_{\tau_2+\delta_2})^+\mathbf{1}_{\{\tau_2+\delta_2<
T\}}.
$$
Hence, from~\eqref{optimal_stopping_delay_special_3}, we know that $Y_t^{\delta_1}\geq Y_t^{\delta_2}$ and $V^{\delta}(\cdot,\cdot)$ is decreasing with respect to $\delta$. }

\textcolor{red}{For the second inequality in~\eqref{bound}, for any
$\tau\in \mathcal{R}_t^{0}$, take
$\widehat{\tau}=(\tau+\delta)\wedge T$. Note that $\widehat{\tau}\in
{\cal F}_{\tau}$ and $\widetilde{X}_t=e^{-(r-q)t}X_t$ is a
martingale. Hence,
    \bee
      &&\mathbf{E}\left(e^{-r\widehat{\tau}}\left(K-X_{\widehat{\tau}}\right)^+\Big|{\cal F}_t\right)
      =\mathbf{E}\left(e^{-r\widehat{\tau}}\left(K-e^{(r-q)\widehat{\tau}}\widetilde{X}_{\widehat{\tau}}\right)^+\Big|{\cal F}_t\right)
      \\[2mm]
      &\geq&\mathbf{E}\left(e^{-r\widehat{\tau}}\left(K-e^{(r-q)\widehat{\tau}}\widetilde{X}_{\tau}\right)^+\Big|{\cal F}_t\right)
      \geq\mathbf{E}\left(\left(e^{-r\widehat{\tau}}K-e^{-q\tau}\widetilde{X}_{\tau}\right)^+\Big|{\cal F}_t\right),
    \eee
where the first inequality follows from the facts that
$e^{-r\widehat{\tau}}\left(K-e^{(r-q)\widehat{\tau}}x\right)^+$ is
convex with respect to $x$ and measurable with respect to ${\cal
F}_\tau$, so we may take conditional expectation with respect to
${\cal F}_\tau$ and apply Jensen's inequality. For the second
inequality, we have used the facts that $q\geq0$ and
$\widehat{\tau}\geq\tau$. In turn,
    \bee
      &&\mathbf{E}\left(e^{-r\tau}\left(K-X_{\tau}\right)^+\Big|{\cal F}_t\right)
      \\[2mm]
      &=&\mathbf{E}\left(\left(e^{-r\widehat{\tau}}K-e^{-q\tau}\widetilde{X}_{\tau}\right)^+
      +\left[e^{-r\tau}\left(K-X_{\tau}\right)^+
      -\left(e^{-r\widehat{\tau}}K-e^{-r\tau}X_{\tau}\right)^+\right]\Big|{\cal F}_t\right)
      \\[2mm]
      &\leq&\mathbf{E}\left(e^{-r\widehat{\tau}}\left(K-X_{\widehat{\tau}}\right)^+\Big|{\cal F}_t\right)+K\mathbf{E}\left(e^{-r\tau}-e^{-r\widehat{\tau}}\Big|{\cal F}_t\right)
      \\[2mm]
      &\leq&\mathbf{E}\left(e^{-r((\tau+\delta)\wedge T)}\left(K-X_{(\tau+\delta)\wedge T}\right)^+\Big|{\cal F}_t\right)
      +K\left(1-e^{-r\delta}\right),
    \eee
where we have used the above conclusion and the fact that
$(x+y)^+-x^+\leq y^+$ in the first inequality, and $\tau\geq0$ and
$\widehat{\tau}-\tau\leq\delta$ in the second inequality. Until now,
we have proved that for any $\tau\in \mathcal{R}_t^{0}$, \bee
   &&\mathbf{E}\left[e^{-r(\tau-t)}(K-X_{\tau})^+\mathbf{1}_{\{\tau<
T\}}+\ e^{-r(T-t)}(K-X_{T})^+\mathbf{1}_{\{\tau\geq
T\}}|\mathcal{F}_t\right]
   \\[2mm]
   &\leq&\mathbf{E}\left[e^{-r(\tau+\delta-t)}(K-X_{\tau+\delta})^+\mathbf{1}_{\{\tau+\delta<
T\}}+\ e^{-r(T-t)}(K-X_{T})^+\mathbf{1}_{\{\tau+\delta\geq
T\}}|\mathcal{F}_t\right]
\\[2mm]
   & &+K\left(1-e^{-r\delta}\right).
\eee Thus, from~\eqref{optimal_stopping_delay_special_3}, we obtain
the second inequality in~\eqref{bound}.}

\textcolor{red}{Finally, we prove the following
inequality~\eqref{derivative with respect to t}, which is important
to analyze the properties of the optimal exercise boundary later on.
\begin{equation}\label{derivative with respect to t}
{\partial_{t}V^{\delta}\leq 0 \;\;\mbox{a.e.
in}\;\;\Omega_{T-\delta}.}
\end{equation}}
\textcolor{red}{By the Markov property and time homogeneity, it is
clear that \bee V^\delta(t,x)=
\esssup_{\tau^{0}\in\mathcal{R}_0^{0}}
\mathbf{E}\left[e^{-r((\tau^{0}+\delta)\wedge(T-t))}(K-X^{0,x}_{(\tau^{0}+\delta)\wedge(T-t)})^+|\mathcal{F}_t\right],\;\;
(t,x)\in \Omega_{T-\delta}, \eee where the notation $X^{0,x}$ means
the state process $X$ starts at the initial time $0$ and position
$x$}

\textcolor{red}{Let $0\leq t_1<t_2\leq T-\delta$. For any $\tau_2\in
\mathcal{R}_0^{0}$, take $\tau_1=\tau_2\wedge(T-\delta-t_2)$. It is
not difficult to check that $\tau_1\in \mathcal{R}_0^{0}$ and
$$
 (\tau_1+\delta)\wedge(T-t_1)
 =(\tau_2+\delta)\wedge(T-t_2).
$$
Thus, we deduce that $V^\delta(t_1,x)\geq V^\delta(t_2,x)$ and
$V^\delta$ is non-increasing with respect to $t$, which further
implies \eqref{derivative with respect to t}.}

\subsection{An early exercise premium decomposition formula}

We derive a decomposition formula for the American put option with
delivery lags.
Such a decomposition formula is crucial to the
analysis of the optimal exercise boundary in sections
\ref{sec:perpetual} and \ref{sec:finite_horizon}.
%
Let $U^{\delta}(t,X)=Y^{\delta}(t,X)-P(T-\delta,X)$. Then, we deduce that
$U^{\delta}(t,X)$ satisfies the variational
inequality \be \label{VI1}
 \left\{
 \begin{array}{ll}
 (-\p_t-{{\cal L}})U^{\delta}(t,X)=\Theta^{\delta}(X),\ &\mbox{if}\;U^{\delta}(t,X)>0,\;\mbox{for}\;(t,X)\in{\Omega}_{T-\delta};
 \vspace{2mm} \\
 (-\p_t-{{\cal L}})U^{\delta}(t,X)\geq \Theta^{\delta}(X),\ &\mbox{if}\;U^{\delta}(t,X)=0,\;\mbox{for}\;(t,X)\in{\Omega}_{T-\delta};
 \vspace{2mm}\\
 U^{\delta}(T-\delta,X)=0,\ &\mbox{for}\ X\in\mathbb{R}_+,
 \end{array}
 \right.
\ee
where $\Theta^{\delta}(\cdot)$ is the Greek
Theta of the European option:
$$\Theta^{\delta}(X)=-\partial_tP(T-\delta,X).$$ Interestingly, we observe that the above variational inequality (\ref{VI1}) also corresponds to an auxiliary optimal stopping problem
\begin{equation}\label{optimal_stopping_delay_special_4}
U^{\delta}(t,X_t)=\esssup_{\tau^{0}\in\mathcal{R}_t^{0}}\mathbf{E}\left[\int_t^{\tau^{0}}e^{-r(s-t)}\Theta^{\delta}(X_s)ds|\mathcal{F}_t\right].
\end{equation}
with its optimal stopping time $\tau^{0,*}$ given in (\ref{optimal_stopping_rule}). In turn, we obtain a decomposition formula for the American put option with delivery lags
\begin{equation}\label{Decomposition}
Y^{\delta}(t,X_t)=P(T-\delta,X_t)+U^{\delta}(t,X_t),
\end{equation}

\begin{remark}\label{remark}
One advantage of the optimal stopping formulation
(\ref{optimal_stopping_delay_special_4}) is that it does not have
final payoff but only has running payoff, and this will facilitate
our analysis of the associated optimal exercise boundary. In the
rest of the paper, we shall focus our analysis on the optimal
stopping problem (\ref{optimal_stopping_delay_special_4}) and its
associated variational inequality (\ref{VI1}).
\end{remark}

To solve (\ref{VI1}), introduce the transformation\footnote{For notation simplicity,
we suppress the superscript $\delta$ in $u^{\delta}$ and
$\theta^{\delta}$, and use $u$ and $\theta$ instead. The same
convention applies to the optimal exercise boundary $x(\tau)$ in
section \ref{sec:finite_horizon}.}
\begin{equation}\label{transform}
 x=\ln X-\ln K,\quad \tau=T-\delta-t,\quad
 u(\tau,x)=U^{\delta}(t,X),\quad
 \theta(x)=\Theta^{\delta}(X).
\end{equation}
Consequently, (\ref{VI1}) reduces to \be  \label{VI2}
 \left\{
 \begin{array}{ll}
 (\p_\tau -\widetilde{{\cal L}})u(t,x)=\theta(x),\ &\mbox{if}\;u(\tau,x)>0,\;\mbox{for}\;(\tau,x)\in{\cal N}_{T-\delta};
 \vspace{2mm} \\
 (\p_\tau -\widetilde{{\cal L}})u(t,x)\geq \theta(x),\ &\mbox{if}\;u(\tau,x)=0,\;\mbox{for}\;(\tau,x)\in{\cal N}_{T-\delta};
 \vspace{2mm}\\
 u(0,x)=0,\ &\mbox{for}\ x\in\mathbb{R},
 \end{array}
 \right.
\ee
 where $
 {\cal N}_{T-\delta}=(0,T-\delta\,]\times\mathbb{R}$, and
 \begin{equation*}
 \widetilde{{\cal
 L}}={\sigma^2\over2}\,\p_{xx}+\left(\,r-q-{\sigma^2\over2}\,\right)\p_x-r.
\end{equation*}
Moreover, it follows from Proposition \ref{regularity1} that $u\in W^{2,1}_{p,loc}({\cal N}_{T-\delta})\cap
C(\overline{{\cal N}_{T-\delta}})$ for $p\geq
  1$ and $\p_x u\in C(\overline{{\cal N}_{T-\delta}})$.

For the latter use, we present some basic properties of the Greek
$\Theta^{\delta}(X)$ whose proof is given in Appendix \ref{appendix:2}.

\begin{proposition}\label{Pro3} Let $\theta(x)=\Theta^{\delta}(X)$ with $x=\ln
X-\ln K$.
Then, the following assertions hold:

{(i)} There exists a unique
zero crossing point $\overline{X}\in\mathbb{R}$ such that
 $\theta(\overline{X})=0$.
 In addition, $\theta(x)<0$ for any $x<\overline{X}$, $\theta(x)>0$ for any
 $x>\overline{X}$, and $\theta^{\prime}(\overline{X})>0$.

 {(ii) For any $x<\overline{X}$, $\theta(x)\rightarrow qKe^x-rK$ as $\delta\rightarrow0^+$.}
\end{proposition}

\section{The perpetual case and its optimal exercise
boundary}\label{sec:perpetual}

We consider the perpetual version of the optimal stopping problem
(\ref{optimal_stopping_delay_special_4}), whose solution admits
explicit expressions (cf. (\ref{ivisolution}) and (\ref{l0}) below).
The perpetual problem is also closely related to the asymptotic
analysis of the optimal exercise boundary in section
\ref{sec:finite_horizon}.

For any
$\mathbb{F}$-stopping time $\tau^{0}\geq t$, we consider the perpetual version of (\ref{optimal_stopping_delay_special_4}), i.e.
\begin{equation}\label{optimal_stopping_delay_special_5}
U^{\delta}_{\infty}(X_t)=\esssup_{\tau^{0}\geq
t}\mathbf{E}\left[\int_t^{\tau^0}e^{-r(s-t)}\Theta^{\delta}(X_s)ds|\mathcal{F}_t\right].
\end{equation}
Using the similar arguments as in section \ref{sec:american}, we
obtain that $U^{\delta}_{\infty}(X)=u_{\infty}(x)$, where $x=\ln
X-\ln K$, and $u_{\infty}(\cdot)$ is the unique {bounded} strong
solution to the stationary variational inequality \be\label{IVI}
  \left\{
 \begin{array}{ll}
 -\widetilde{{\cal L}}\,u_\infty(x)=\theta(x),
 &\mbox{if}\;u_\infty(x)>0,\;\mbox{for}\;x\in\mathbb{R};
 \vspace{2mm} \\
 -\widetilde{{\cal L}}\,u_\infty(x)\geq \theta(x),
 &\mbox{if}\;u_\infty(x)=0,\;\mbox{for}\;x\in\mathbb{R},
 \end{array}
 \right.
\ee with $u_\infty\in W^2_{p,loc}(\mathbb{R})$ for
 $p\geq1$ and $(u_\infty)^\prime\in C(\mathbb{R})$.

\textcolor{red}{From Proposition \ref{Pro3}, we know that $\{\theta(x)\geq 0\}=\{x\geq \overline{X}\}$. In this domain, we consider the following PDE,
\begin{equation}\label{v infty}
 -\widetilde{{\cal L}}\,v_\infty(x)=\theta(x)>0,\;x\in(\overline{X},+\infty)\qquad
 v_\infty(\overline{X})=0.
\end{equation}
The above PDE has a unique classical solution $v\in C^2(\overline{X},+\infty)\cap C[\overline{X},+\infty)$. The strong maximum principle (see \cite{Evans}) implies that $v>0$ in $(\overline{X},+\infty)$.}

\textcolor{red}{Moreover, it is clear that $u_\infty$ satisfies
$$
 -\widetilde{{\cal L}}\,u_\infty(x)\geq \theta(x),\;x\in(\overline{X},+\infty)\qquad
 u_\infty(\overline{X})\geq0.
$$
Using the comparison principle (see {\cite{Friedman2} or \cite{Yan1}}) for the strong solution of PDE in $(\overline{X},+\infty)$, we deduce that $u_\infty\geq v_\infty>0$ in $(\overline{X},+\infty)$. So, it follows that
\begin{equation}\label{region}
\{x> \overline{X}\}\subseteq\{u_{\infty}(x)>0\}\ \ \text{and}\ \
\{x\leq \overline{X}\}\supseteq\{u_{\infty}(x)=0\}.
\end{equation}}

We can then define the \emph{optimal exercise boundary}
$\underline{X}$ as\footnote{{Note that from the definition of
$\underline{X}$, it may be possible that $\underline{X}=-\infty$. We
will however exclude such a situation in Proposition \ref{le1}.}}
\begin{equation}\label{freeboundary_0}
\underline{X}=\inf\{x\in\mathbb{R}: u_{\infty}(x)>0\}.
\end{equation}
The continuity of $u_{\infty}(\cdot)$ implies that $u_{\infty}(x)=0$
for $x\leq \underline{X}$ and, therefore, the player will exercise
the option in $(-\infty,\underline{X}]$. Moreover, it follows from
(\ref{region}) and (\ref{freeboundary_0}) that $\underline{X}\leq
\overline{X}$.

The next proposition relates variational inequality (\ref{IVI}) to a
free-boundary problem, which in turn provides the explicit
expressions for $u_{\infty}(\cdot)$ and $\underline{X}$.

\begin{proposition}\label{le1}
For $x>\underline{X}$, it holds that $u_{\infty}(x)>0$. Moreover,
$(u_{\infty}(\cdot),\underline{X})$ is the unique {bounded} solution
to the free-boundary problem \be\label{free_boundary_problem_1}
  \left\{
 \begin{array}{ll}
 -\widetilde{{\cal L}}\,u_\infty(x)=\theta(x),
 &\mbox{for}\;x>\underline{X};
 \vspace{2mm} \\
 u_\infty(x)=0,
 &\mbox{for}\;x\leq \underline{X};
 \vspace{2mm} \\
 (u_{\infty}){'}(\underline{X})=0,&\text{(smooth-pasting
 condition)},
 \end{array}
 \right.
\ee and satisfies {$\overline{X}> \underline{X}>-\infty$}.
\end{proposition}

\begin{proof} {\emph{Step 1}. We prove that
$(u_{\infty}(\cdot),\underline{X})$ satisfies the free-boundary
problem (\ref{free_boundary_problem_1}).} To this end, we first show
what $u_{\infty}(x)>0$ for $x>\underline{X}$. Since
$u_{\infty}(x)>0$ for $x>\overline{X}$, we only need to show that
$u_{\infty}>0$ on $(\underline{X},\overline{X}]$. If not, let
$x_1,\,x_2\in[\,\underline{X},\overline{X}\,]$ be such that
$$x_1<x_2,\,u_{\infty}(x_1)=u_{\infty}(x_2)=0,\ \text{and}\
u_{\infty}(x)>0\ \text{for\ any}\ x\in(x_1,x_2).$$ Using variational
inequality (\ref{IVI}) and Proposition \ref{Pro3}, we obtain that
\begin{eqnarray*}
  \left\{
 \begin{array}{ll}
 -\widetilde{{\cal L}}\,u_\infty(x)=\theta(x)\leq 0,
 &\mbox{for}\ x\in(x_1,x_2);
 \vspace{2mm} \\
 u_\infty(x_1)=u_{\infty}(x_2)=0.
 &
 \end{array}
 \right.
\end{eqnarray*}
The comparison principle then implies that $u_{\infty}(x)\leq 0$ for
$x\in(x_1,x_2)$, which is a contradiction.

To prove the smooth-pasting condition, we observe that
$(u_{\infty})^{\prime}$ is continuous, and that $u_{\infty}(x)=0$
for $x\leq \underline{X}$.  Therefore,
$(u_{\infty})^{\prime}(\underline{X}+0)=(u_{\infty})^{\prime}(\underline{X}-0)=0$,
and $(u_{\infty}(\cdot),\underline{X})$ indeed satisfies the free
boundary problem (\ref{free_boundary_problem_1}).\smallskip

\emph{Step 2.} we prove that $(u_{\infty}(\cdot),\underline{X})$ is
actually the unique solution to (\ref{free_boundary_problem_1}). To
this end, we first show that if
$(u_{\infty,1}(\cdot),\underline{X}_1)$ is any solution solving
(\ref{free_boundary_problem_1}), then it is necessary that
$\underline{X}_1<\overline{X}$. If not, by
(\ref{free_boundary_problem_1}) and Proposition \ref{Pro3}, we have
\begin{eqnarray*}
  \left\{
 \begin{array}{ll}
 -\widetilde{{\cal L}}\,u_{\infty,1}(x)=\theta(x)>0,
 &\mbox{for}\;x>\underline{X}_1\geq \overline{X};
 \vspace{2mm} \\
 u_{\infty,1}(\underline{X}_1)=(u_{\infty,1}){'}(\underline{X}_1)=0.
 \end{array}
 \right.
\end{eqnarray*}
The strong comparison principle (see \cite{Evans}) then implies that
$u_{\infty,1}(x)>0$ for $x>\underline{X}_1$.

Next we compare $u_{\infty,1}(x)$ with an auxiliary function
$$\underline{w}(x)=u_{\infty,1}(\underline{X}_1+1)w(x;\underline{X}_1,\underline{X}_1+1)$$
in the interval $(\underline{X}_1,\underline{X}_1+1)$, where
$$
 w(x;a,b)={e^{\lambda^+(x-a\,)}-e^{\lambda^-(x-a\,)}\over e^{\lambda^+(b-a)}-e^{\lambda^-(b-a)}},
$$ with $\lambda^+$ and $\lambda^-$ being, respectively, the positive and negative
characteristic roots of $\widetilde{\cal{L}}$:
$$
 {\sigma^2\over 2}\lambda^2+\left(\,r-q-{\sigma^2\over
 2}\,\right)\lambda-r=0.
$$
It is clear that
$$w(a;a,b)=0,\quad w(b;a,b)=1,\quad w'(a;a,b)>0,\quad -\widetilde{\cal{L}}w=0\ \text{in}\ (a,b).$$
In turn,
\begin{eqnarray*}
  \left\{
 \begin{array}{ll}
  -\widetilde{{\cal
 L}}\underline{w}(x)=0<-\widetilde{{\cal L}}u_{\infty,1}(x),\qquad\mbox{for}\;x\in(\underline{X}_1,\underline{X}_1+1);
 \vspace{2mm} \\
 u_{\infty,1}(\underline{X}_1)=\underline{w}(\underline{X}_1),\qquad\ \

 u_{\infty,1}(\underline{X}_1+1)=\underline{w}(\underline{X}_1+1).
 \end{array}
 \right.
\end{eqnarray*}
Hence, the comparison principle implies that $
 u_{\infty,1}(x)\geq \underline{w}(x)$ for
 $x\in(\underline{X}_1,\underline{X}_1+1)$.
In turn, $(u_{\infty,1})^\prime(\underline{X}_1)\geq
 \underline{w}^\prime(\underline{X}_1)>0,
$
 which contradicts the smooth-pasting condition $(u_{\infty,1})^\prime(\underline{X}_1)=0$.

Now we show that $(u_{\infty}(\cdot),\underline{X})$ is the unique
solution to (\ref{free_boundary_problem_1}). If not, let
 $(u_{\infty,\,1},\underline{X}_1)$ be another solution of the free-boundary problem \eqref{free_boundary_problem_1}. Without loss of
 generality, we may assume that
 $\underline{X}_1<\underline{X}<\overline{X}$.
It is immediate to check that
\begin{eqnarray*}
 \left\{
 \begin{array}{ll}
  -\widetilde{\cal L}u_{\infty,\,1}(x)=\theta(x)\\
  \ \ \ \ \ \ \ \ \ \ \ \ \ \ \ \ \leq \theta(x) I_{\{x>\underline{X}\}}
 =-\widetilde{\cal L}u_{\infty}(x),\ \text{for}\
 x\in(\,\underline{X}_1,\infty);
 \vspace{2mm}\\
 u_{\infty,\,1}(\underline{X}_1)=u_{\infty}(\underline{X}_1)=0;\\
 (u_{\infty,\,1})^\prime(\underline{X}_1)=(u_{\infty})^\prime(\underline{X}_1)=0,
 \end{array}
 \right.
\end{eqnarray*}
 where we have used the fact $\theta(x)<0$ for any $x\leq \underline{X}<\overline{X}$. {The comparison principle} then implies that
 $u_{\infty,\,1}(x)\leq u_{\infty}(x)$ and, in particular,
 $u_{\infty,\,1}(x)\leq u_{\infty}(x)=0$ for
 $x\in[\,\underline{X}_1,\underline{X}]$.

On the other hand, applying Taylor's expansion to $u_{\infty,1}(x)$
yields
$$u_{\infty,1}(x)=\frac12u_{\infty,1}''(\underline{X}_1+0)(x-\underline{X}_1)^2(1+o(1))=\frac{-\theta(\underline{X}_1)}{\sigma^2}
(x-\underline{X}_1)^2(1+o(1)),$$ which further implies that
$u_{\infty,1}(x)>0$ if $x$ is close enough to $\underline{X}_{1}$.
Thus, we obtain a contradiction.

\smallskip

\emph{Step 3}. We prove that  $\overline{X}> \underline{X}>-\infty$.
Since we have already showed that $\underline{X}<\overline{X}$ in
Step 2, it is sufficient to prove that $\underline{X}>-\infty$.

{In fact, using} the free-boundary formulation
(\ref{free_boundary_problem_1}), we further obtain that its solution
must have the form
\begin{equation*}
 u_\infty(x)=CK e^{\lambda^-x}-p(x),\ \text{for}\
 x>\underline{X},
\end{equation*}
where the constants $C$ is to be determined,
 $\lambda^-$ is the negative root of the characteristic equation
 for $\widetilde{\cal{L}}$,
and
 $p(x)=p(T-\delta,x)$ is the price of the European put option (cf.
(\ref{price_european}) with $t=T-\delta$).

In order to fix the constant $C$ and  the optimal exercise boundary
$\underline{X}$, we make use of the boundary and smooth-pasting
conditions in (\ref{free_boundary_problem_1}), and obtain that
\begin{eqnarray*}
 \left\{
 \begin{array}{ll}
 CKe^{\lambda^-\underline{X}}=p(\underline{X}\,)=\left[\,Ke^{-r\delta}N(-\underline{d}\,_2)
 -Ke^{\underline{X}-q\delta}N(-\underline{d}\,_1)\right];
 \vspace{2mm}\\
 CK\lambda^-e^{\lambda^-\underline{X}}=p^\prime(\underline{X}\,)
 =-Ke^{\underline{X}-q\delta}N(-\underline{d}\,_1),
 \end{array}
 \right.
\end{eqnarray*} where $\underline{d}\,_1$ and $\underline{d}\,_2$ are the same
as $d_1$ and $d_2$ in (\ref{defofN}) except that $x$ is replaced by
$\underline{X}$ (see Appendix \ref{appendix:2} for the notations).
Thus, we obtain that \be\label{ivisolution}
 u_\infty(x)=\left\{
 \begin{array}{ll}
 p(\underline{X}) \,e^{\lambda^-(x-\underline{X}\,)}-p(x) ,
 &\text{for}\ x>\underline{X}\;;
 \vspace{2mm}\\
 0 &\text{for}\ x\leq\underline{X}\,,
 \end{array}
 \right.
\ee and $\underline{X}$ is the zero crossing point of the algebraic
equation \be\label{l0}
 l(x)=\lambda^-e^{-r\delta}N(-{d}_2)+(1-\lambda^-)e^{x-q\delta}N(-{d}_1)=0.
\ee

{Next, we prove that the zero crossing point of $l(x)=0$ exists and
is unique. It is clear that, when $x\rightarrow-\infty$,
$$
 d_1,\,d_2\rightarrow-\infty,\;\;N(-d_1),\,N(-d_2)\rightarrow1,\;\;
 l(x)\rightarrow\lambda^-e^{-r\delta}+o(1)<0.
$$
Hence, $l(x)$ is negative provided $x$ is small enough. On the other
hand, by (\ref{inequ1}) and (\ref{inequ2}), we have
$$
 d_1,\,d_2\rightarrow+\infty,\;\;
 N(-d_1)={N'(-d_1)\over d_1}\Big(1+o(1)\Big),\;\;
 N(-d_2)={N'(-d_2)\over d_2}\Big(1+o(1)\Big),
$$
as $x\rightarrow+\infty$, and therefore,
\begin{eqnarray*}
 {l(x)e^{r\delta}\over N'(-d_2)}
 &=&{\lambda^-\over d_2}\Big(1+o(1)\Big)
 +{1-\lambda^-\over d_1}\Big(1+o(1)\Big)
 \\[2mm]
 &=&{d_2+\lambda^-(d_1-d_2)\over
 d_1\,d_2}\Big(1+o(1)\Big)={1\over d_1}\Big(1+o(1)\Big).
\end{eqnarray*}
Hence, $l(x)$ is positive provided $x$ is large enough. Thus, we
deduce that there exists at least one zero crossing point of
$l(x)=0$.} Thanks to the uniqueness of the solution to the
free-boundary problem (\ref{free_boundary_problem_1}), we know that
the zero crossing point of the algebraic equation (\ref{l0}) is also
unique, from which we conclude that $\underline{X}>-\infty$.
\end{proof}

\section{The optimal exercise boundary and its asymptotic
analysis}\label{sec:finite_horizon}

With all the preparations, we are ready to prove Theorem \ref{main_theorem} (ii) and (iii).
For illustration purpose,
we first
demonstrate the optimal exercise boundary through Figures 1 and 2 as
 below.

 \begin{picture}(0,130)(140,0)
\put(190,10){\vector(1,0){150}} \put(260,6){\vector(0,1){100}}
\put(258,108){$\tau$}\put(343,10){$x$}
{\thicklines\qbezier(280,9)(220,46)(210,100)}
\put(207,10){\line(0,1){10}}\put(207,25){\line(0,1){10}}\put(207,40){\line(0,1){10}}
\put(207,55){\line(0,1){10}}\put(207,70){\line(0,1){10}}\put(207,85){\line(0,1){10}}
 \put(277,7){$\bullet$}\put(282,12){$\overline{X}$}
 \put(200,50){${\bf ER}$}\put(251,50){${\bf CR}$}
 \put(200,105){${ x(\tau)}$}
 \put(200,0){$\underline{X}$}
 \put(175,80){$u=0$}\put(250,80){$u>0$}
\end{picture}
\begin{center}
\small{Figure 1: Optimal exercise boundary $x(\tau)$ under the
coordinates $(\tau,x)$.}
\end{center}

 \begin{picture}(0,130)(140,0)
\put(190,10){\vector(1,0){150}} \put(200,6){\vector(0,1){110}}
\put(195,115){$t$}\put(343,10){$X$}
{\thicklines\qbezier(225,10)(240,85)(310,90)}
\put(220,10){\line(0,1){10}}\put(220,25){\line(0,1){10}}\put(220,40){\line(0,1){10}}
\put(220,55){\line(0,1){10}}\put(220,70){\line(0,1){10}}\put(220,85){\line(0,1){10}}
\put(200,90){\line(1,0){10}}\put(215,90){\line(1,0){10}}\put(230,90){\line(1,0){10}}
\put(245,90){\line(1,0){10}}\put(260,90){\line(1,0){10}}\put(275,90){\line(1,0){10}}
\put(290,90){\line(1,0){10}}\put(305,90){\line(1,0){10}}\put(310,90){\line(1,0){8}}
\put(325,90){\line(1,0){10}}
\put(198,87){$\bullet$}\put(174,84){$T-\delta$}
\put(200,105){\line(1,0){10}}\put(215,105){\line(1,0){10}}\put(230,105){\line(1,0){10}}
\put(245,105){\line(1,0){10}}\put(260,105){\line(1,0){10}}\put(275,105){\line(1,0){10}}
\put(290,105){\line(1,0){10}}\put(305,105){\line(1,0){10}}\put(310,105){\line(1,0){8}}
\put(325,105){\line(1,0){10}}
\put(198,102){$\bullet$}\put(182,101){$T$}
 \put(210,50){${\bf ER}$}\put(320,50){${\bf CR}$}
 \put(260,65){${X^\delta(t)}$}
 \put(218,7){$\bullet$}\put(200,-3){$K e^{\underline{X}}$}
 \put(308,87){$\bullet$}\put(300,92){$K e^{\overline{X}}$}
 \put(205,63){$U^\delta=0$}\put(310,63){$U^\delta>0$}
\end{picture}
\begin{center}
\small{Figure 2: Optimal exercise boundary $X^{\delta}(t)$ under the
coordinates $(t,X)$.}
\end{center}

Figure 1 is under the coordinates $(\tau,x)$, and Figure 2 is under
the coordinates $(t,X)$,
 where $\tau=T-\delta-t$ and $x=\ln X-\ln
 K$ (cf. the transformation (\ref{transform})).
Figure 2 illustrates that the whole region $\Omega_{T-\delta}$ is
divided by a curve $X^{\delta}(t)$ into two parts. In the left
region, the investor will exercise the option (with time lag
$\delta$), and in the right region the investor will hold the
option. Hence, $X^{\delta}(t)$ is called the \emph{optimal exercise
boundary}. If we denote by $x(\cdot)$ the optimal
 exercise boundaries under the coordinates $(\tau,x)$, as shown in
 Figure 1, then we have the
 relationship
\begin{equation}\label{relation}
 X^{\delta}(t)=K\exp{\{x(T-\delta-t)\}}.
 \end{equation}
\subsection{Proof of Theorem \ref{main_theorem} (ii)}

Due to Remark \ref{remark}, we will mainly work with variational
inequality (\ref{VI2}) for $u(\cdot,\cdot)$. Recall
${\cal N}_{T-\delta}=(0,T-\delta\,]\times\mathbb{R}$. Define the
exercise domain ${\bf ER}$ and the continuation domain ${\bf CR}$ as
\begin{align*}
 {\bf ER}&=\{(\tau,x)\in\mathcal{N}_{T-\delta}:u(\tau,x)=0\};\qquad
\\{\bf CR}&= \{(\tau,x)\in\mathcal{N}_{T-\delta}:u(\tau,x)>0\}.
\end{align*}

\begin{lemma}\label{Th3}
Let $\overline{X}$ and $\underline{X}$ be given in Proposition
\ref{Pro3} and (\ref{freeboundary_0}), respectively. Then, it holds
that
$$\{\,x\leq \overline{X}\,\}\supseteq{\bf ER}\supseteq\{\,x\leq\underline{X}\,\}\ \ \text{and}\ \ \{\,x> \overline{X}\,\}\subseteq{\bf CR}\subseteq\{\,x>\underline{X}\,\}.$$
\end{lemma}
\begin{proof}
 In order to prove that ${\bf ER}\supseteq\{\,x\leq \underline{X}\,\}$, we compare $u(\cdot,\cdot)$ and $u_\infty(\cdot)$,
 the latter of which is the solution to variational
 inequality~\eqref{IVI}. Note that
\bee
  \left\{
 \begin{array}{ll}
 (\p_\tau -\widetilde{{\cal L}})u_\infty(x)=\theta(x),&\mbox{if}\;u_\infty(x)>0,\;\mbox{for}\;(\tau,x)\in{\cal N}_{T-\delta};
 \vspace{2mm} \\
 (\p_\tau-\widetilde{{\cal L}})u_\infty(x)\geq \theta(x),&\mbox{if}\;u_\infty(x)=0,\;\mbox{for}\;(\tau,x)\in{\cal N}_{T-\delta};
 \vspace{2mm}\\
 u_\infty(x)\geq 0=u(0,x),&\text{for}\ x\in\mathbb{R}.
 \end{array}
 \right.
\eee
 The comparison principle for variational inequality (\ref{VI2}) in the domain $\mathcal{N}_{T-\delta}$ then implies that $u(\tau,x)\leq u_\infty(x)$. But
 if $x\leq \underline{X}$, according to the free-boundary problem (\ref{free_boundary_problem_1}), $u_{\infty}(x)=0$. In turn, $u(\tau,x)=0$. This proves
 that
 $\{x\leq \underline{X}\,\}\subseteq {\bf ER}$.

 \textcolor{red}{Repeating the above argument to compare $u$ and $v_\infty$ in the domain $\{x\geq \underline{X}\,\}$, where $v_\infty$ is the classical solution of PDE~\eqref{v infty}, we obtain $u\geq v_\infty>0$ in the domain $\{x>\underline{X}\,\}$, it follows that
$\{\,x>\overline{X}\,\}\subseteq{\bf CR}$.}
\end{proof}

Intuitively, when $\theta(x)$ is positive (i.e. $x> \overline{X}$),
the running payoff in (\ref{optimal_stopping_delay_special_4}) is
also positive, so the investor will hold the option. In the
contrary, when $\theta(x)$ is non-positive (i.e. $x\leq
\overline{X}$), one may think that the investor would then exercise
the option to stop her losses. However, the above lemma shows that
for {$x\leq\underline{X}$}, the investor may still hold the option,
and wait for the running payoff to rally at a later time to recover
her previous losses.

Next, we define the \emph{optimal exercise boundary} $x(\tau)$ as
\begin{equation}\label{freeboundary_10}
x(\tau)=\inf\{x\in\mathbb{R}: u(\tau,x)>0\},
\end{equation} for any $\tau\in(0,T-\delta]$. It follows from Lemma \ref{Th3}
that $x(\tau)\in[\underline{X},\overline{X}]$, and by the continuity
of $u(\cdot,\cdot)$, $u(\tau,x)=0$ for $x\leq x(\tau)$.

\begin{lemma}\label{lemma11} For $\tau\in(0,T-\delta]$, let
$$x_1(\tau)=\sup\{x\in\mathbb{R}: u(\tau,x)=0\}.$$
Then, $x(\tau)=x_{1}(\tau)$. Hence, $x(\tau)$ is the unique curve
separating $\mathcal{N}_{T-\delta}$ such that $u(\tau,x)=0$ for
$x\leq x(\tau)$ and $u(\tau,x)>0$ for $x\geq x(\tau)$.
\end{lemma}

\begin{proof} The definition of $x_{1}(\tau)$ implies that $x(\tau)\leq
x_1(\tau)$ and $u(\tau,x)>0$ for $x\geq x_1(\tau)$. Moreover, it
follows from Lemma \ref{Th3} that
$x_1(\tau)\in[\underline{X},\overline{X}]$.

Suppose $x(\tau^*)< x_1(\tau^*)$ for some $\tau^*\in(0,T-\delta]$.
The continuity of $u$ implies that
$u(\tau^*,x(\tau^*))=u(\tau^*,x_1(\tau^*))=0$. Let $x^*$ be a
maximum point of $u(\tau^*,\cdot)$ in the interval
$[x(\tau^*),x_1(\tau^*)]$. Suppose that $u(\tau^*,x^*)> 0$;
otherwise {$u(\tau^*,x)\equiv0$ in the interval
$[x(\tau^*),x_1(\tau^*)]$, which contradicts the definition of
$x^*(\tau)$}. Since $u(\tau^*,x^*)>0$, $\partial_xu(\tau^*,x^*)=0$
and $\partial_{xx}u(\tau^*,x^*)\leq0$, we have
$$-\widetilde{{\cal L}}u(\tau^*,x^*)=-{\sigma^2\over2}\,\p_{xx}u(\tau^*,x^*)-\left(\,r-q-{\sigma^2\over2}\,\right)\p_xu(\tau^*,x^*)+ru(\tau^*,x^*)> 0.$$

On the other hand, by the continuity of $u$, there exits a
neighborhood of $(\tau^*,x^*)$ such that $u>0$, so
$\p_{\tau}u-\widetilde{{\cal L}}u=\theta$. In turn,
$$-\widetilde{{\cal L}}u(\tau^*,x^*)=\theta(x^*)-\p_{\tau}u(\tau^*,x^*).$$
Since \begin{equation}\label{timederivative}
\partial_{\tau}u(\tau,x)=-\partial_{t}U^{\delta}(t,X)=-\partial_{t}V^{\delta}(t,X)\geq
0,
\end{equation}
where we have used the transformation (\ref{transform}) and the
decomposition (\ref{Decomposition}) in the first two equalities, and
(\ref{derivative with respect to t}) in the last inequality, we
further get
$$-\widetilde{{\cal L}}u(\tau^*,x^*)\leq \theta(x^*)< 0.$$ This is a
contradiction. Thus, we must have $x(\tau^*)=x_1(\tau^*)$.
\end{proof}

From the above lemma, we deduce that the exercise region and the
continuation region are equivalent to
\begin{align*}
 {\bf ER}&=\{(\tau,x)\in {\cal N}_{T-\delta}:x\leq x(\tau)\};\ \\
 {\bf CR}&=\{(\tau,x)\in {\cal N}_{T-\delta}:x> x(\tau)\}.
\end{align*}

We return to the proof of Theorem \ref{main_theorem} (ii). Note that it is equivalent to the following proposition in terms of $x(\cdot)$.

\begin{proposition}\label{Th5}
Let $x(\tau)$ be the optimal exercise boundary given in
(\ref{freeboundary_10}). Then, the following assertions {hold}:

(i) Monotonicity: ${x(\tau)}$ is strictly decreasing in $\tau$;
\footnote{\textcolor{red}{Recently, \cite{TA} provides a new
probabilistic argument to prove that the free boundary is strictly
monotonic. We thank the referee for pointing out \cite{TA}.}}

(ii) Position: $x(\tau)$ is with
 the starting point
 $x(0)=\lim\limits_{\tau\rightarrow0^+}x(\tau)=\overline{X}$;

(iii) Regularity: ${x(\cdot)}\in C^\infty(0,T-\delta\,]$ and
${u(\cdot,\cdot)}\in C^\infty(\{x\geq x(\tau):\tau\in(0,T-\delta\,]\}).$
\end{proposition}
\begin{proof}
 (i) We first show that $x(\tau)$ is non-increasing.  For any
$0\leq \tau_1<\tau_2\leq T-\delta$, we then have
$0=u(\tau_2,x(\tau_2))\geq u(\tau_1,x(\tau_2))\geq 0$. Thus,
$u(\tau_2,x(\tau_2))=u(\tau_1,x(\tau_2))=0$, and together with Lemma
\ref{lemma11}, we deduce that $x(\tau_1)\geq x(\tau_2)$, i.e.
$x(\tau)$ is non-increasing.

 If $x(\tau)$ is not strictly decreasing, then there exist
 $x_1\in[\underline{X},\overline{X}]$ and
 $0\leq\tau_1<\tau_2\leq T-\delta$ such that $x(\tau)=x_1$ for any
 $\tau\in[\tau_1,\tau_2]$. See Figure 3 below.
Note that $\partial_xu(\tau,x_1)=0$ and, moreover,
$\partial_{\tau}\partial_{x}u(\tau,x_1)=0$ for any
$\tau\in[\tau_1,\tau_2]$.

On the other hand, we observe that in the domain
$[\tau_1,\tau_2]\times(x_1,x_1+1)$, $u(\cdot,\cdot)$ satisfies
\begin{eqnarray*}
  \left\{
 \begin{array}{ll}
 (\partial_{\tau}-\widetilde{{\cal L}})u(\tau,x)=\theta(x),
 &\mbox{for}\ (\tau, x)\in[\tau_1,\tau_2]\times(x_1,x_1+1);
 \vspace{2mm} \\
 u(\tau,x_1)=0,
 &\mbox{for}\ \tau\in[\tau_1,\tau_2].
 \end{array}
 \right.
\end{eqnarray*}
In turn, $\partial_{\tau}u(\cdot,\cdot)$ satisfies
\begin{eqnarray*}
  \left\{
 \begin{array}{ll}
 (\partial_{\tau}-\widetilde{{\cal L}})\partial_{\tau}u(\tau,x)=\partial_{\tau}\theta(x)= 0,
 &\mbox{for}\ (\tau, x)\in[\tau_1,\tau_2]\times(x_1,x_1+1);
 \vspace{2mm} \\
 \partial_{\tau}u(\tau,x_1)=0,
 &\mbox{for}\ \tau\in[\tau_1,\tau_2].
 \end{array}
 \right.
\end{eqnarray*}
For any $x_2>\overline{X}$, since $(\tau_2,x_2)\in\mathbf{CR}$, we
have $u(\tau_2,x_2)>0$, and $u(0,x_2)=0$. Hence, there exists
$\tau\in(0,\tau_2)$ such that $\p_\tau u(\tau,x_2)>0$. Note,
however, that $\partial_{\tau}u\geq 0$ (cf. (\ref{timederivative}))
and, therefore, the strong maximum principle (see \cite{Evans})
implies that $\p_\tau u>0$ in ${\bf CR}$.

Together with $\partial_{\tau}u(\tau,x_1)=0$ for any
$\tau\in[\tau_1,\tau_2]$, we deduce that
$\partial_x\partial_{\tau}u(\tau,x_1)>0$ using Hopf lemma (see
\cite{Evans}). But this is a contradiction to
$\partial_{\tau}\partial_{x}u(\tau,x_1)=0$ for any
$\tau\in[\tau_1,\tau_2]$.

(ii) It is obvious that $x(0)\leq \overline{X}$ from Lemma
\ref{Th3}, so it is sufficient to prove that $x(0)\geq
\overline{X}$. If not, in the domain
$(0,T-\delta]\times(x(0),\overline{X})\subset\mathbf{CR}$, we
consider
\begin{eqnarray*}
  \left\{
 \begin{array}{ll}
 (\partial_{\tau}-\widetilde{{\cal L}})u(\tau,x)=\theta(x)<0,
 &\mbox{for}\ (\tau, x)\in(0,T-\delta]\times(x(0),\overline{X});
 \vspace{2mm} \\
 u(0,x)=0,
 &\mbox{for}\ x\in(x(0),\overline{X}).
 \end{array}
 \right.
\end{eqnarray*}
Then,
$\partial_{\tau}u(0,x)=\widetilde{\cal{L}}u(0,x)+\theta(x)=\theta(x)<0$,
which is a contradiction to $\partial_{\tau}u\geq 0$ in
(\ref{timederivative}).

(iii) We first prove that $x(\tau)$ is continuous. If not, then
there exists $\tau_2\in(0,T-\delta)$ and $\underline{X}\leq
x_3<x_1\leq \overline{X}$ such that $x(\tau_2+0)=x_3$ and
$x(\tau_2-0)=x_1$. See Figure 3 below.

In the domain $(\tau_2,T-\delta]\times(x_3,x_1)\subset\mathbf{CR}$,
we consider
\begin{eqnarray*}
  \left\{
 \begin{array}{ll}
 (\partial_{\tau}-\widetilde{{\cal L}})u(\tau,x)=\theta(x)<0,
 &\mbox{for}\ (\tau, x)\in[\tau_2,T-\delta]\times(x_3,x_1);
 \vspace{2mm} \\
 u(\tau_2,x)=0,
 &\mbox{for}\ x\in(x_3,x_1).
 \end{array}
 \right.
\end{eqnarray*}
Then,
$\partial_{\tau}u(\tau_2,x)=\widetilde{\cal{L}}u(\tau_2,x)+\theta(x)=\theta(x)<0$,
which is a contradiction to $\partial_{\tau}u\geq 0$ in
(\ref{timederivative}).

Finally, since $\partial_{\tau}u\geq 0$, the smoothness of both the
optimal exercise boundary $x(\tau)$ and the value function
$u(\cdot,\cdot)$ in the continuation region follow along the similar
arguments used in \cite{Fr2}.
\end{proof}

\begin{picture}(0,130)(140,0)
 \put(240,10){\vector(1,0){140}} \put(360,6){\vector(0,1){100}}
 \put(363,100){$\tau$}\put(383,10){$x$}
 {\thicklines\qbezier(355,10)(340,15)(330,30)
 \put(300,40){\line(1,0){30}}
 \qbezier(300,40)(282,42)(275,50)
 \qbezier(275,50)(265,60)(260,92)}
 \multiput(300,40)(0,-5){6}{\line(0,-1){2}}
 \multiput(330,40)(0,-5){6}{\line(0,-1){2}}
 \put(328,7){$\bullet$}\put(328,-3){$x_1$}
 \put(252,7){$\bullet$}\put(250,-3){$\underline{X}$}
 {\thicklines\put(330,30){\line(0,1){10}}}
 \put(366,7){$\bullet$}\put(366,-3){$x_2$}
 \put(297,7){$\bullet$}\put(298,-3){$x_3$}
 \put(353,7){$\bullet$}\put(353,-3){$\overline{X}$}
 \multiput(330,40)(5,0){7}{\line(1,0){2}}
 \multiput(255,10)(0,5){17}{\line(0,1){2}}
 \multiput(330,30)(5,0){7}{\line(1,0){2}}
 \put(358,37){$\bullet$}\put(368,37){$\tau_2$}
 \put(358,27){$\bullet$}\put(368,27){$\tau_1$}
 \put(240,60){${\bf ER}$}\put(235,72){$u=0$}\put(321,60){${\bf CR}$}\put(317,72){$u>0$}
 \put(280,50){$x(\tau)$}
 \end{picture}
 \begin{center}
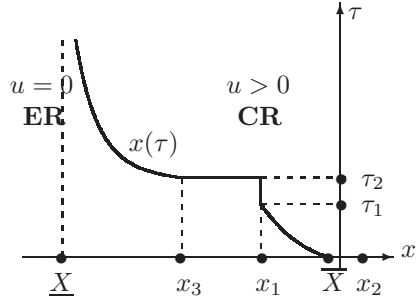
 \small{Figure 3: Non-strictly decreasing and discontinuous free boundary
 $x(\tau)$.}
 \end{center}

%
%
%

\subsection{Proof of Theorem \ref{main_theorem} (iii): Asymptotic behavior for large maturity}

We study the asymptotic behavior of the optimal exercise boundary
$x(\tau)$ and the value function $u(\tau,x)$ as
$\tau\rightarrow\infty$, which in turn proves Theorem \ref{main_theorem} (iii) for the asymptotic behavior of $X^{\delta}(t)$ when $T\rightarrow \infty$.

To this end, we consider the auxiliary optimal stopping time problem
perturbed by $r\ep$,
\begin{equation}\label{optimal_stopping_delay_special_5.1}
U^{\epsilon}_{\infty}(X_t)=\esssup_{\tau^{0}\geq
t}\mathbf{E}\left[\int_t^{\tau^0}e^{-r(s-t)}(\Theta(X_s)-r\epsilon)ds|\mathcal{F}_t\right],
\end{equation}
for any $\mathbb{F}$-stopping time $\tau^{0}\geq t$ and any
$\epsilon\geq 0$. This will help us to achieve the lower bound and,
therefore, the asymptotic behavior of the optimal exercise boundary
$x(\tau)$.

Following along the similar arguments used in section
\ref{sec:perpetual}, we obtain that
$u^{\epsilon}_{\infty}(x)=U^{\epsilon}_{\infty}(X)$, where $x=\ln
X-\ln K$, and $u^{\ep}(\cdot)$ is the unique strong solution to the
stationary variational inequality \be\label{IVI.1}
  \left\{
 \begin{array}{ll}
 -\widetilde{{\cal L}}\,u^\ep_\infty(x)=\theta(x)-r\ep,
 &\mbox{if}\;u^\ep_\infty(x)>0,\;\mbox{for}\;x\in\mathbb{R};
 \vspace{2mm} \\
 -\widetilde{{\cal L}}\,u^\ep_\infty(x)\geq \theta(x)-r\ep,
 &\mbox{if}\;u^\ep_\infty(x)=0,\;\mbox{for}\;x\in\mathbb{R},
 \end{array}
 \right.
\ee with $u^\ep_\infty\in W^2_{p,loc}(\mathbb{R})$ for
 $p\geq1$ and $(u^\ep_\infty)^\prime\in C(\mathbb{R})$.

In contrast to variational inequality (\ref{IVI}), it is not clear
how to reduce variational inequality (\ref{IVI.1}) to a
free-boundary problem, and to obtain its explicit solution.
Nevertheless, we are able to derive a local version of the
free-boundary problem with $\epsilon>0$ small enough, which is
sufficient to obtain the asymptotic behavior of the optimal exercise
boundary later on.

\begin{lemma}\label{Pro5}
 For $\ep>0$ small enough, it holds that $u_{\infty}(x)\geq u^\epsilon_\infty(x)\geq
 u_\infty(x)-\epsilon$. Define $\underline{X}^{\epsilon}$ as
 \begin{equation*}
\underline{X}^{\epsilon}=\inf\{x\in(-\infty,\overline{X}]:
u^{\epsilon}_{\infty}(x)>0\}.
\end{equation*}
Then $\underline{X}\leq \underline{X}^{\epsilon}< \overline{X}$, and
$u^\epsilon_\infty(x)>0$ for any
$x\in(\underline{X}^{\epsilon},\overline{X})$, where $\underline{X}$
and $\overline{X}$ are given in (\ref{freeboundary_0}) and
Proposition \ref{Pro3}, respectively. Moreover,
$\underline{X}^\epsilon\rightarrow \underline{X}$ as
$\ep\rightarrow{0^+}$. {See Figure 4 below.}

\vspace{-1cm}

\begin{picture}(0,130)(140,0)
\put(150,20){\vector(1,0){200}}\put(350,10){$x$}
\put(300,65){$\theta(x)$}
{\thicklines\qbezier(170,2)(220,24)(260,60)}
{\thicklines\qbezier(260,60)(280,75)(295,65)}
{\thicklines\qbezier(295,65)(320,46)(355,39)}
\put(204,17){$\bullet$}\put(204,7){$\overline{X}$}
\put(185,17){$\bullet$}\put(185,25){$\underline{X}^\epsilon$}
\put(155,17){$\bullet$}\put(155,25){$\underline{X}$}
\end{picture}
\begin{center}
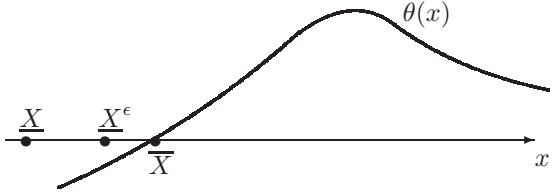
\small{Figure 4: The graph of $\underline{X}$,
$\underline{X}^{\epsilon}$ and $\overline{X}$.}
\end{center}
\end{lemma}

\begin{proof} Note that the running payoff in the optimal stopping problem
(\ref{optimal_stopping_delay_special_5.1}) satisfies
\begin{align*}
\int_t^{\tau^0}e^{-r(s-t)}\Theta(X_s)ds&\geq
\int_t^{\tau^0}e^{-r(s-t)}(\Theta(X_s)-r\epsilon)ds\\
&=\int_t^{\tau^0}e^{-r(s-t)}\Theta(X_s)ds+\epsilon
e^{-r(\tau^0-t)}-\epsilon \\
&\geq \int_t^{\tau^0}e^{-r(s-t)}\Theta(X_s)ds-\epsilon,
\end{align*}
 for any $\mathbb{F}$-stopping time
$\tau^0\geq t$. It follows that $u_{\infty}(x)\geq
u^\epsilon_\infty(x)\geq u_\infty(x)-\epsilon$.

Since $u_{\infty}(x)>0$ for $x>\underline{X}$, and
$\overline{X}>\underline{X}$ by Proposition \ref{le1}, it holds that
$u_{\infty}(\overline{X})>0$. Let $\epsilon>0$ be small enough such
that $\epsilon<u_{\infty}(\overline{X})$. Using the inequality
$u^{\epsilon}_{\infty}(x)\geq u_{\infty}(x)-\epsilon$, we obtain
that
$$u_{\infty}^{\epsilon}(\overline{X})\geq
u_{\infty}(\overline{X})-\epsilon>0.$$ In turn, the definition of
$\underline{X}^{\epsilon}$ and the continuity of
$u^{\epsilon}_{\infty}(\cdot)$ imply that
$\underline{X}^{\epsilon}<\overline{X}$.

Repeating the similar arguments used in Proposition \ref{le1}, we
obtain that $u^{\epsilon}_\infty(x)>0$ for
$x\in(\underline{X}^{\epsilon},\overline{X})$. Furthermore, the
inequality $u_{\infty}(x)\geq u^\epsilon_\infty(x)$ and Proposition
\ref{le1} imply that
$$
0=u_{\infty}(\underline{X})\geq u_{\infty}^{\ep}(\underline{X}).
$$
In turn, the definition of $\underline{X}^{\epsilon}$ implies that
$\underline{X}^{\epsilon}\geq \underline{X}$.

{Next, we prove that
$\underline{X}^{\epsilon}\rightarrow\underline{X}$ as
$\epsilon\rightarrow0^+$. In fact, from the definition of
$\underline{X}^\epsilon$ and the continuity of $u^\epsilon_\infty$,
we know that $u^\epsilon_\infty(\underline{X}^\epsilon)=0$. Using
the inequality $u^{\epsilon}_{\infty}(x)\geq u_{\infty}(x)-\epsilon$
again, we obtain $u_{\infty}(\underline{X}^\epsilon)\leq\epsilon$.}

{On the other hand, applying Taylor's expansion to $u_{\infty}(x)$
yields
$$u_{\infty}(x)=\frac12u_{\infty}''(\underline{X}+0)(x-\underline{X})^2(1+o(1))=\frac{-\theta(\underline{X})}{\sigma^2}
(x-\underline{X})^2(1+o(1)),$$ which further implies that
$u_{\infty}(x)>\kappa (x-\underline{X})^2$ with some positive
constant $\kappa$ if $x$ is close enough to $\underline{X}$.
Moreover, since $u_{\infty}(x)>0$ in the interval
$(\underline{X},\overline{X}\,]$ and is continuous, we deduce that
if $\epsilon$ is small enough, then $\underline{X}^\epsilon\leq
\underline{X}+\sqrt{\epsilon/\kappa}$. Recalling
$\underline{X}^\epsilon\geq \underline{X}$, we conclude that
$\underline{X}^{\epsilon}\rightarrow\underline{X}$ as
$\epsilon\rightarrow0^+$.}
\end{proof}

We return to the proof of Theorem \ref{main_theorem} (iii), which is equivalent to the following proposition.

\begin{proposition}
\label{theorem} Let $u(\cdot,\cdot)$ and $x(\tau)$ be the solution
to variational inequality (\ref{VI2}) and its associated optimal
exercise boundary (\ref{freeboundary_10}), respectively. Then,
$$u(\tau,\cdot)\rightarrow u_\infty(\cdot)\ \ \text{and}\ \ x(\tau)\rightarrow
\underline{X},$$
 as $\tau\rightarrow\infty$, where $u_\infty(\cdot)$ and $\underline{X}$ are
 the solution of the stationary
 variational inequality (\ref{IVI}) and its associated optimal exercise boundary (\ref{freeboundary_0}), respectively,
\end{proposition}

\begin{proof}
From the optimal stopping problems
(\ref{optimal_stopping_delay_special_4}) and
(\ref{optimal_stopping_delay_special_5}), it is immediate that
$u(\cdot,\cdot)\leq u_{\infty}(\cdot)$. For $t\leq (T-\delta)/2$,
define
$$
 {u}^t(\tau,x)=u^{\exp\{-rt\}}_\infty(x)-e^{-r(\tau-t)}+e^{-rt},
$$
 where ${u}^{\exp\{-rt\}}_\infty(\cdot)$ is the solution of variational inequality~\eqref{IVI.1}
 with $\ep=\exp\{-rt\}$. It is routine to check that
 ${u}^t\in W^{2,\,1}_{p,\,loc}({\cal N}_{2t})\cap C(\overline{{\cal
 N}_{2t}}\,)$, and satisfies
\bee
  \left\{
 \begin{array}{ll}
 (\p_{\tau}-\widetilde{{\cal L}}){u}^t(\tau,x)=\theta(x),
 \ \ \ \mbox{if}\;{u}^t(\tau,x)>-e^{-r(\tau-t)}+e^{-rt},\;\mbox{for}\;(\tau,x)\in{\cal N}_{2t};
 \vspace{2mm} \\
 (\p_{\tau} -\widetilde{{\cal L}}){u}^t(\tau,x)\geq \theta(x),\ \ \
 \mbox{if}\;{u}^t(\tau,x)=-e^{-r(\tau-t)}+e^{-rt},\;\mbox{for}\;(\tau,x)\in{\cal N}_{2t};
 \vspace{2mm} \\
 {u}^t(0,x)={u}^{\exp\{-rt\}}_\infty(x)-e^{rt}+e^{-rt}<0,\ \ \
 \mbox{for}\ x\in\mathbb{R},
 \end{array}
 \right.
\eee
 provided that $t$ and $T$ are large enough. Since the obstacle
 $-e^{-r(\tau-t)}+e^{-rt}\leq0$ in the domain ${\cal N}_{2t}$, using the
 comparison principle (see
\cite{Friedman2} or \cite{Yan1}) for variational inequality
(\ref{VI2}) in the domain ${\cal N}_{2t}$, we
 deduce that
$
 u(\tau,x)\geq{u}^t(\tau,x)$ for $(\tau,x)\in{\cal N}_{2t}$. In turn, Lemma
 \ref{Pro5} implies that
\begin{equation}\label{inequ3}
 u(2t,\cdot)\geq {u}^t(2t,\cdot)={u}^{\exp\{-rt\}}_\infty(\cdot)
 \geq {u}_\infty(\cdot)-e^{-rt}.
\end{equation}
Together with $u(2t,\cdot)\leq u_{\infty}(\cdot)$, we obtain that
$u(2t,\cdot)\rightarrow u_{\infty}(\cdot)$ as $t\rightarrow\infty$.

To prove the convergence of the optimal exercise boundary $x(\tau)$
to $\underline{X}$, we choose $t$ large enough such that
$\underline{X}^{\exp\{-rt\}}+\exp\{-rt\}<\overline{X}$. Then,
(\ref{inequ3}) yields that
$$
 u\left(2t,\underline{X}^{\exp\{-rt\}}+\exp\{-rt\}\right)\geq{u}^{\exp\{-rt\}}_\infty\left(\underline{X}^{\exp\{-rt\}}+\exp\{-rt\}\right)>0,$$
where we have used ${u}^{\exp\{-rt\}}_\infty\left(x\right)>0$ for
$x\in(\underline{X}^{\exp\{-rt\}},\overline{X})$ (cf. Lemma
\ref{Pro5}) in the second inequality. It then follows from the
definition of $x(\tau)$ in (\ref{freeboundary_10}) that $$x(2t)\leq
\underline{X}^{\exp\{-rt\}}+\exp\{-rt\}.$$ By Lemma \ref{Th3}, we
also have $x(\tau)\geq
 \underline{X}$ for any $\tau\in[\,0,T-\delta\,].
$ Hence, we have proved that $$\underline{X}\leq x(2t)\leq
 \underline{X}^{\exp\{-rt\}}+\exp\{-rt\}.
$$
Finally, we send $t\rightarrow\infty$ in the above inequalities, and
conclude the convergence of $x(2t)$ to $\underline{X}$ by
Lemma \ref{Pro5}.
\end{proof}

\begin{remark} Under the original coordinates $(t,X)$, it follows
from the relationship (\ref{relation}) and Proposition \ref{theorem}
that $X^{\delta}(t)\rightarrow K e^{\underline{X}}$ as
$T\rightarrow\infty$, so $Ke^{\underline{X}}$ is the asymptotic line
of the optimal exercise boundary $X^{\delta}(t)$.

Proposition \ref{theorem} also establishes the connection between the
optimal stopping problems (\ref{optimal_stopping_delay_special_4})
and (\ref{optimal_stopping_delay_special_5}):
$U^{\delta}(t,X)\rightarrow U^{\delta}_{\infty}(X)$ uniformly in
$X\in\mathbb{R}_+$ as $T\rightarrow\infty$. Moreover, it follows
from the decomposition formula (\ref{Decomposition}) that the value
function of the American put option with time lag $\delta$ has the
long maturity limit: $V^{\delta}(t,X)\rightarrow
P(T-\delta,X)+U_{\infty}^{\delta}(X)$ uniformly in
$X\in\mathbb{R}_+$ as $T\rightarrow\infty$.
\end{remark}

\subsection{Proof of Theorem \ref{main_theorem} (iii): Asymptotic behavior for small time lag}

Finally, we prove Theorem \ref{main_theorem} (iii) for the asymptotic behavior of $X^{\delta}(t)$ when $\delta\rightarrow 0$. Recall that $X^0(t)$ denotes the optimal exercise boundary of the
corresponding standard American put option. It is well known that
$X^{0}(t)$ is a strictly increasing and smooth function with
$X^0(T)=K$. We refer to \cite{Chen} and \cite{Jiang} for its proof.

We first extend variational inequality (\ref{VI11}) from
$\Omega_{T-\delta}$ to $\Omega_T$ by defining
$V^{\delta}(t,X)=P(t,X)$ for
$(t,X)\in[T-\delta,T]\times\mathbb{R}_+$, and rewrite (\ref{VI11})
as
\begin{align}\label{PDE}
(-\p_t -{\cal L})V^\delta
(t,X)&=I_{\{V^\delta=P(T-\delta,X)\}}(-\p_t-{\cal L})P(T-\delta,X)\notag\\
&=-I_{\{V^\delta=P(T-\delta,X)\}}\Theta(X),
\end{align}
for $(t,X)\in\Omega_{T}$, and $V^{\delta}(T,X)=(K-X)^+$ for
$X\in\mathbb{R}^+$.


Denote $\mathcal{N}_T^{n}:=(0,T]\times\mathcal{N}^n$ and
$\mathcal{N}^n:=(-n,K-\frac1n)$. Then, we apply the
$W^{2,1}_p$-estimates (see Lemma A.4 in \cite{Yan1} for example) to
the above PDE (\ref{PDE}) for $V^{\delta}(\cdot,\cdot)$, and obtain
that for any $n\in\mathbb{N}$, \be\label{estimate}
 \|V^\delta\|_{W^{2,1}_p(\mathcal{N}_T^{n})}\leq
 C\Big(\,\|V^\delta\|_{L^p(\mathcal{N}_T^{2n})}+\|\Theta\|_{L^p(\mathcal{N}^{2n})}+\|K-X\|_{W^{2,1}_p(\mathcal{N}^{2n})}
\Big). \ee Note that the right hand side of the above inequality is
independent of $\delta$ due to the fact that $V^{0}(t,X)-\textcolor{red}{K(1-e^{-r\delta})}\leq V^{\delta}(t,X)\leq V^{0}(t,X)$ (cf. (\ref{bound})), and the
formula (\ref{freeterm}) for $\Theta(X)$.



From Theorem \ref{main_theorem} (i), $V^\delta$ converges to $V^0$ in
$C(\overline{\Omega_T})$ as $\delta\rightarrow 0$.  Hence, the above
estimate~\eqref{estimate} implies that $V^\delta$ also converges
weakly to $V^0$ in $W^{2,1}_p(\mathcal{N}^n_T)$ and
\begin{equation*}
 -I_{\{V^\delta=P(T-\delta,X)\}}\Theta(X)=(-\p_t -{\cal L})V^{\delta}(t,X)\rightharpoonup
(-\p_t-{\cal
 L})V^0(t,X)
\end{equation*}
 weakly in $L^p(\mathcal{N}^n_T)$ as $\delta\rightarrow 0$. But note
 that $$(-\p_t-{\cal
 L})V^0(t,X)=I_{\{V^0=K-X\}}(rK-qX).$$ In turn,
\begin{equation}\label{convergence}
-I_{\{V^\delta=P(T-\delta,X)\}}\Theta(X)\rightharpoonup
I_{\{V^0=K-X\}}(qX-rK)
\end{equation}
weakly in $L^p(\mathcal{N}^n_T)$.

Now suppose that $X^{\delta}(t)$ does not converge to $X^{0}(t)$.
Then there exist $t_0\in[0,T)$ and a sequence
$\{X^{\delta_{m}}\}_{m=1}^{\infty}$ such that when
$\delta_m\rightarrow 0$, $X^{\delta_{m}}(t_0)$ does not converge to
$X^{0}(t_0)$.

Since $X^{0}(t)$ is continuous and strictly increasing with
$X^{0}(T)=K$, we may assume there exists $\epsilon>0$ and an integer
$M$ such that $X^{0}(t_0)+2\epsilon<\min\{X^{\delta_m}(t_0),K\}$ for
any $m\geq M$. See Figure 5 below. Other cases can be treated in a
similar way.

By the continuity and strictly increasing property of both
$X^{0}(t)$ and $X^{\delta}(t)$, we can find $\eta>0$ such that the
compact set $[t_{0},t_0+\eta]\times
[X^0(t_0)+\epsilon,X^0(t_0)+2\epsilon]$ is in the exercise region of
$V^{\delta_m}$ and the continuation region of $V^0$. Therefore, in
this compact set,
$V^{\delta_m}(t,X)=P(T-\delta_m,X),\,V^0(t,X)>K-X$, and
\begin{align*}
&-I_{\{V^{\delta_m}=P(T-\delta_m,X)\}}\Theta(T-\delta_m,X)-I_{\{V^0=K-X\}}(qX-rK)\\
=&-\Theta(T-\delta_m,X),
\end{align*}
where we use the notation $\Theta(T-\delta_m,\cdot)$ to emphasize
its dependence on $T-\delta_m$. However, from Proposition
\ref{Pro3}, it is immediate to check that
$$
 \lim_{\delta_m\rightarrow
 0}\Theta(T-\delta_m,X)=qX-rK<0,{\;\text{for}\;X<K,}
$$
 which is a contradiction to (\ref{convergence}).

\begin{picture}(0,130)(140,0)
\put(190,10){\vector(1,0){150}} \put(200,6){\vector(0,1){110}}
\put(195,115){$t$}\put(343,10){$X$}
{\thicklines\qbezier(225,10)(226,60)(310,90)}
{\thicklines\qbezier(230,10)(235,85)(312,105)}
\put(200,90){\line(1,0){10}}\put(215,90){\line(1,0){10}}\put(230,90){\line(1,0){10}}
\put(245,90){\line(1,0){10}}\put(260,90){\line(1,0){10}}\put(275,90){\line(1,0){10}}
\put(290,90){\line(1,0){10}}\put(305,90){\line(1,0){10}}\put(310,90){\line(1,0){8}}
\put(325,90){\line(1,0){10}}
\put(198,87){$\bullet$}\put(174,84){$T-\delta$}
\put(200,105){\line(1,0){10}}\put(215,105){\line(1,0){10}}\put(230,105){\line(1,0){10}}
\put(245,105){\line(1,0){10}}\put(260,105){\line(1,0){10}}\put(275,105){\line(1,0){10}}
\put(290,105){\line(1,0){10}}\put(305,105){\line(1,0){10}}\put(310,105){\line(1,0){8}}
\put(325,105){\line(1,0){10}}
 \put(200,75){\line(1,0){10}}\put(215,75){\line(1,0){10}} \put(230,75){\line(1,0){10}}
 \put(245,75){\line(1,0){5}}\put(255,75){\line(1,0){20}}
 \put(267,82){\line(1,0){6}}\put(267,75){\line(0,1){7}}\put(273,75){\line(0,1){7}}
\put(198,102){$\bullet$}\put(182,101){$T$}
\put(198,72){$\bullet$}\put(190,71){$t_0$}
 \put(210,35){${\bf ER}$}\put(320,35){${\bf CR}$}
 \put(260,56){${X^\delta(t)}$}\put(218,56){${X^0(t)}$}
 \put(308,87){$\bullet$}\put(300,92){$K e^{\overline{X}}$}
 \put(308,102){$\bullet$}\put(300,107){$K$}
\end{picture}
\begin{center}
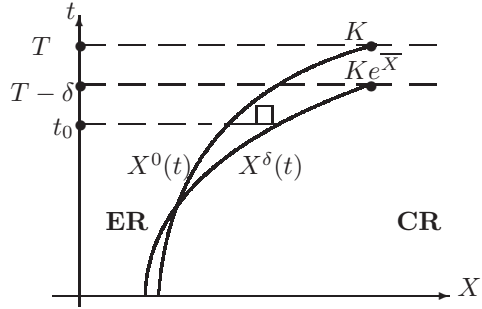
 \small{Figure 5: Non-convergence of the free boundaries $X^{\delta}(t)$ to
 $X^{0}(t)$ as $\delta\rightarrow 0$.}
 \end{center}

\section{Conclusion}

This paper studies the asymptotic behavior of the value function and
the optimal exercise boundary of American put options with delivery
lags through free boundary techniques. On one hand, it would be
interesting to carry out the free boundary analysis to the real
option setup such as reversible investment (\cite{Bar-ILan2 },
\cite{Bar-ILan}, \cite{Cost}), impulse control (\cite{Erhan},
\cite{Pham}, \cite{Oksendal2}), and recursive optimal stopping
(\cite{TA0}, \cite{TA2}). \textcolor[rgb]{1.00,0.00,0.00}{On the
other hand, it might be possible to prove the convexity of the
optimal exercise boundary (as in \cite{Jiang_1} and \cite{Ekstrom}
for the standard American put case).} Such extensions are left for
the future research.

\appendix

\section{Proof of Proposition \ref{Pro3}}\label{appendix:2}
{(i)} We first show that the function $\theta(x)$ (or equivalently,
$\Theta(X)$ with $X=Ke^{x}$) has the explicit form
\be\label{freeterm}
 \theta(x)=qKe^{x-q\delta}N(-d_1)+{\sigma K\over 2\sqrt{\delta}}\,e^{-r\delta}N^{\prime}(-d_2)-rKe^{-r\delta}N(-d_2),
\ee
 where $N(d)={1\over\sqrt{2\pi}}\int_{-\infty}^d
 e^{-\xi^2\over2}d\xi$,\
 $N^{\prime}(d)={1\over\sqrt{2\pi}}e^{\frac{-d^2}{2}}$, and
\begin{equation}\label{defofN}
 d_1={x\over\sigma\sqrt{\delta}}
 +\left(\,{r-q\over\sigma}+{\sigma\over2}\right)\,\sqrt{\delta},\qquad
 d_2=d_1-\sigma\sqrt{\delta}.
\end{equation}

To this end, let $x=\ln X-\ln K$ and $p(t,x)=P(t,X)$. It is well
known that $p(t,x)$ has the explicit expression (see \cite{Jiang}
for example)
\begin{equation}\label{price_european}
 p(t,x)=Ke^{-r(T-t)}N(-{d}_2^{t})-Ke^{x-q(T-t)}N(-{d}_1^t),
\end{equation}
 where $d^t_1$ and $d_t^{2}$ are the same as $d_1$ and $d_2$ in (\ref{defofN}) except that $\delta$ is
 replaced $T-t$:
$$
 {d}_1^t={x\over\sigma\sqrt{T-t}}
 +\left(\,{r-q\over\sigma}+{\sigma\over2}\right)\,\sqrt{T-t},\quad
 {d}_2^t={d}_1^t-\sigma\sqrt{T-t}.
$$

Differentiating $p(t,x)$ against $t$ yields that
\begin{align*}\nonumber
 \p_tp(t,x)
 =&\ rKe^{-r(T-t)}N(-{d}_2^t)-qKe^{x-q(T-t)}N(-{d}_1^t)
 \\[2mm]\nonumber
 &\ -Ke^{-r(T-t)}N^{\prime}(-{d}_2^t)\left(\p_t{d}_1^t+{\sigma\over
 2\sqrt{T-t}}\right)+Ke^{x-q(T-t)}N^{\prime}(-{d}_1^t)\,\p_t{d}_1^t
 \\[2mm]\label{pderivative}
 =&\ rKe^{-r(T-t)}N(-{d}_2^t)-qKe^{x-q(T-t)}N(-{d}_1^t)\notag\\
 &-{\sigma K\over 2\sqrt{T-t}}\,e^{-r(T-t)}N^{\prime}(-{d}_2^t),
\end{align*}
 where
we have used the fact that
\begin{equation}\label{equ1}
 e^{-r(T-t)}N^{\prime}(-{d}_2^t)=e^{x-q(T-t)}N^{\prime}(-{d}_1^t).
\end{equation}
Thus we have proved (\ref{freeterm}).

To prove Proposition \ref{Pro3} (i), we use the following two
elementary inequalities: For $d\geq 0$,
\begin{align}
N(-d)&\textcolor{red}{<}
\frac{1}{\sqrt{2\pi}}\int_{-\infty}^{-d}e^{-\frac{\xi^2}{2}}\frac{\xi}{-d}d\xi=\frac{1}{\sqrt{2\pi}d}e^{-\frac{d^2}{2}}=\frac{1}{d}N^{\prime}(-d);\label{inequ1}\\
N(-d)&\textcolor{red}{>}
\frac{1}{\sqrt{2\pi}}\int_{-\infty}^{-d}e^{-\frac{\xi^2}{2}}\frac{1+\frac{1}{\xi^2}}{1+\frac{1}{d^2}}d\xi=\frac{1}{\sqrt{2\pi}(d+\frac{1}{d})}e^{-\frac{d^2}{2}}=\frac{1}{d+\frac{1}{d}}N^{\prime}(-d)\label{inequ2}.
\end{align}

We first show that there exists $\overline{X}$ such that
$\theta(\overline{X})=0$. Note that this is equivalent to show that
${\theta(\overline{X})}/{N^{\prime}(-\overline{d}_2)}=0$, where
$\overline{d}_2$ is the same as $d_2$ in (\ref{defofN}) except that
$x$ is replaced by $\overline{X}$.

For $x$ large enough such that $d_1,d_2\geq 0$ (cf. (\ref{defofN})),
we have
\begin{align*}
\frac{\theta(x)}{N^{\prime}(-d_2)}&=
 qKe^{x-q\delta}\frac{N(-d_1)}{N^{\prime}(-d_2)}
 +\frac{\sigma K}{2\sqrt{\delta}}e^{-r\delta}
 -rKe^{-r\delta}\frac{N(-d_2)}{N^{\prime}(-d_2)}\\
 &\geq
 qKe^{x-q\delta}\frac{1}{d_1+\frac{1}{d_1}}\frac{N^{\prime}(-d_1)}{N^{\prime}(-d_2)}
 +\frac{\sigma K}{2\sqrt{\delta}}e^{-r\delta}
 -rKe^{-r\delta}\frac{1}{d_2}\frac{N^{\prime}(-d_2)}{N^{\prime}(-d_2)},
\end{align*}
by using the inequalities (\ref{inequ1}) and (\ref{inequ2}). From
(\ref{equ1}), we further obtain that
\begin{equation*}
\frac{\theta(x)}{N^{\prime}(-d_2)}\geq
  \frac{qK}{d_1+\frac{1}{d_1}}e^{-r\delta}+
  \frac{\sigma K}{2\sqrt{\delta}}e^{-r\delta}-
  rKe^{-r\delta}\frac{1}{d_2}>0,
\end{equation*}
provided that $d_2\geq 2r\sqrt{\delta}/\sigma$, so
$\frac{\theta(x)}{N^{\prime}(-d_2)}>0$ for large enough $x$.

On the other hand, when $x\rightarrow-\infty$, we have that
$d_1,d_2\rightarrow-\infty$ and, therefore,
$$N(-d_1),\ N(-d_2)\rightarrow 1,\ \text{and}\ N^{\prime}(-d_2)\rightarrow 0.$$
Hence, $\theta(x)\rightarrow-rKe^{-r\delta}<0$. This means that
$\theta(x)$ is negative provided $x$ is small enough, so
$\frac{\theta(x)}{N^{\prime}(-d_2)}<0$ for small enough $x$. Since
$\frac{\theta(x)}{N^{\prime}(-d_2)}$ is obviously continuous in $x$,
we conclude that there exists $\overline{X}\in\mathbb{R}$ such that
${\theta(\overline{X})}/{N^{\prime}(-\overline{d}_2)}=0$.

Next, we show that $\frac{\theta(x)}{N^{\prime}(-d_2)}$ is strictly
increasing in $x$, so its zero crossing point $\overline{X}$ is
unique. Indeed, note that
\begin{equation*}
\left(\frac{\theta(x)}{N^{\prime}(-d_2)}\right)^{\prime}=\frac{Ke^{-r\delta}}{\sigma\sqrt{\delta}}
\left[r-q+q\frac{N(-d_1)d_1}{N^{\prime}(-d_1)}-r\frac{N(-d_2)d_2}{N^{\prime}(-d_2)}\right].
\end{equation*}
Let $h(d):=\frac{N(-d)d}{N^{\prime}(-d)}$. Then, we calculate its
derivative against $d$ as
\begin{equation*}
h^{\prime}(d)=-d+\frac{N(-d)}{N^{\prime}(-d)}+\frac{N(-d)}{N^{\prime}(-d)}d^2.
\end{equation*}
It is obvious that $h^{\prime}(d)>0$ when $d\leq 0$. For $d>0$, by
using the inequalities (\ref{inequ1}) and (\ref{inequ2}), we obtain
that
\begin{equation*} h^{\prime}(d)\textcolor{red}{>}
-d+\frac{1}{d+\frac{1}{d}}+\frac{1}{d+\frac{1}{d}}d^2=0.
\end{equation*}
In turn,  $h(d_2)\textcolor{red}{<}h(d_1)$, which yields that
\begin{equation*}
\left(\frac{\theta(x)}{N^{\prime}(-d_2)}\right)^{\prime}\textcolor{red}{>}
\frac{Ke^{-r\delta}}{\sigma\sqrt{\delta}}
\left[r-q+(q-r)h(\textcolor{red}{d_1})\right]\textcolor{red}{\geq}0
\end{equation*}
by noting that $h(d_2)<\lim_{d\rightarrow\infty}h(d)=1$ and $r>q$.

{(ii) For any $x<0$, since $\delta\rightarrow0^+$,
$d_1,d_2\rightarrow-\infty$, and, therefore,
$$N(-d_1),\ N(-d_2)\rightarrow 1,\ \text{and}\ N^{\prime}(-d_2)\rightarrow 0,\ \theta(x)\rightarrow qKe^x-rK. $$}

\end{document}